\documentclass[11pt]{article}
\usepackage{amsfonts}
\usepackage[
letterpaper=true,colorlinks=true,pdfpagemode=none,
urlcolor=blue,linkcolor=blue,citecolor=blue,pdfstartview=FitH]{hyperref}

\usepackage{amsmath,amsfonts}
\usepackage{graphicx}
\usepackage{color}

\newif\ifblog
\newif\iftex
\blogfalse \textrue

\usepackage{ulem}
\def\em{\it}
\def\emph#1{\textit{#1}}

\newtheorem{theorem}{Theorem}[section]
\newtheorem{lemma}[theorem]{Lemma}

\newtheorem{proposition}[theorem]{Proposition}

\newtheorem{remark}[theorem]{Remark}
\newenvironment{proof}{\noindent {\sc Proof:}}{$\Box$} 

\newcommand{\ep}{\varepsilon}

\newcommand{\s}{\sigma}
\newcommand{\p}{\partial}

\newcommand{\de}{\delta}
\newcommand{\be}{\begin{eqnarray}}
\newcommand{\ee}{\end{eqnarray}}
\newcommand{\beq}{\begin{equation}}
\newcommand{\eeq}{\end{equation}}
\newcommand{\bee}{\begin{eqnarray*}}
\newcommand{\eee}{\end{eqnarray*}}

\title{American option
  of stochastic volatility model with negative Fichera function on degenerate boundary}

\author{Xiaoshan Chen \thanks{Department of Mathematics,
    City University of Hong Kong, {\tt
      xchen53@gmail.com}. } \and
  Qingshuo Song \thanks{Department of Mathematics, City University of
    Hong Kong, {\tt songe.qingshuo@cityu.edu.hk}.}}

\begin{document}
\date{}
\maketitle

\begin{abstract}
  In this paper we study a general framework of American put option with stochastic volatility whose value function is associated with a 2-dimensional parabolic variational inequality  with degenerate boundaries. We  apply PDE methods to analyze the existences of the strong solution and the properties of the 2-dimensional manifold for the
free boundary.  Thanks to the regularity result
on the solution of the underlying PDE,
we can also provide the uniqueness of the solution by the argument
of the verification theorem together with the generalized Ito's formula
even though the solution may not be second order differentiable in the
space variable across the free boundary.

\vskip 0.2 true in \noindent {\bf Keyword:} Stochastic volatility,
American option, verification theorem,
optimal stopping.

\vskip 0.2 true in \noindent {\bf AMS Subject Classification
Number.} 91Gxx, 60G40, 60J67.

\end{abstract}

\newpage

\setlength{\baselineskip}{0.22in}
\section{Introduction}
Option pricing is one of the most important topics in the quantitative finance research. Although the Black-Scholes model has been well studied, empirical evidence suggests that the Black-Scholes model is inadequate to describe asset returns and the behavior of the option markets. One possible remedy is to assume that the volatility of the asset price also follows a stochastic process, see \cite{He93} and \cite{HW87}.

In the standard Black-Scholes model, a standard logarithmic change of variables transforms the Black-Scholes equation into an equation with constant coefficients which can be
studied by a general PDE theory directly.
Different from the standard Black-Scholes PDE,  the
general PDE methods does not directly apply to the
PDE  associated with the option pricing underlying
the stochastic volatility model in the following cases:
1)  The pricing equation is degenerate at the boundary;
2) The drift and volatility coefficients may grows faster than linear, see  \cite{He93}.

In the related literatures, \cite{He93} derived a closed-form solution
for the price of a European call option on some specific stochastic volatility
models. \cite{ET10} also
studied the Black-Scholes equation of European option underlying stochastic volatility models, and showed that the value function was the unique classical solution to a PDE with a certain boundary behavior. Also, \cite{BKX12} showed a necessary and sufficient condition on the uniqueness of classical solution to the valuation PDE of European option in a general framework of stochastic volatility models. 
In contrast to the European option pricing on the stochastic volatility model,
although there have been quite a few approximate solutions
 and numerical approaches, such as \cite{AGG10, CKMZ09},
the study of the existence and uniqueness of strong solution for PDE related to American option price 
on the stochastic volatility model is rather limited. In particular,
the unique solvability of PDE associated with American options of finite maturity with the presence of degenerate boundary and
super-linear growth
has not been studied in an appropriate Sobolev space.

 Note that, American call options with no dividend is equivalent to
 the European call option. For this reason,
 we only consider a general framework of American put option model with stochastic volatility whose  value function is associated by a 2-dimensional parabolic variational inequality. On the other hand, 
in the theory of linear PDE,
boundary conditions along degenerate boundaries should not needed
if the Fichera function is nonnegative, otherwise it should be imposed, see
 \cite{OR73}. Therefore, we only consider the case when Fichera function is negative on the degenerate boundary $y=0$ in this paper, and leave the other case in the future study.

To resolve solvability issue, we adopt similar methodology of
\cite{CSYY13} to work on a PDE of truncated version backward  in time using appropriate penalty function and mollification. The main difference
is that \cite{CSYY13} studies constant drift and volatility, while the current
paper considers functions of drift and volatility and the negative Fichera 
function plays crucial role in the proof.

Uniqueness
issue is usually tackled by comparison result implied by Ishii's lemma
with notions of viscosity solution, see \cite{CIL92}. However, this approach
does not apply in this problem due to the fast growth of drift and
volatility functions on unbounded domain. The approach to establish 
uniqueness in our paper is similar to the classical verification theorem
conducted to classical PDE solution. In fact, a careful construction  leads
to a local regularity of the solutions in Sobolev space, and this
enables us to apply generalized Ito's formula (see \cite{Kry80})  
with weak derivatives. Note that this approach not only provides 
uniqueness of strong solution of PDE, but also provides that 
the value function of American option is exactly the unique strong 
solution.

In the next section, we first introduce the generalized stochastic volatility model.  Section \ref{sec:strongsol} shows the existence of strong solution to the truncated version of the variational inequality. We characterize the free boundary in section \ref{sec:freeboundary}. Section \ref{sec:uniqueness}
shows that the value function of the American option price is the unique  strong solution to the variational inequality with appropriate boundary datum. Concluding remarks are given in section \ref{sec:conclusion}.

\section{
Stochastic volatility model}\label{sec:model}

Let $(\Omega, \mathcal{F}, \mathbb{P})$ be a filtered probability space with a filtration $(\mathcal{F}_t)_{t\geq0}$ that satisfies the usual conditions, $W_t$ and $B_t$ be two standard Brownian motions with correlation $\rho$. Suppose the stock
price follows
$${\rm (Stk)} \quad d X_s = X_s(rds+ \sigma(Y_s) d W_s),\quad X_t = x > 0,$$
 and  volatility follows
$${\rm (Vol)} \quad d Y_s = \mu (Y_s) d s + b(Y_s) d B_s, \quad Y_t = y >0.$$
 Let $X^{x,t}$ and $Y^{y,t}$ be dynamics satisfying (Stk) and (Vol) with respective initial conditions on superscripts.

We consider an American put option underlying the asset $X_s$ with strike $K>0$ and maturity $T$, which has the payoff $(K-X_\tau)^+$ at the
exercise time $\tau\in\mathcal T_{t,T}$. Here
$\mathcal T_{t,T}$ denotes the set of all stopping times in $[t,T]$.
Define value function of the optimal stopping by
 \beq\label{eq:defV}
 V(x,y,t) =\sup\limits_{\tau\in\mathcal T_{t,T}} \mathbb{E}_{x,y,t} [e^{-r(\tau-t)}(K-X_\tau)^+], \quad
(x,y,t) \in \bar{\mathbb{R}}^+\times\bar{\mathbb{R}}^+\times[0,T].
 \eeq
 Following assumptions will be imposed:
\begin{itemize}
\item [(A1)]
$\mu, \sigma^2, b^2$ are locally Lipschitz continuous on $\mathbb R$
with $\mu(0) = \sigma(0) = b(0)=0$ and $\sigma(y), b(y)>0,
\s'(y)\geq0$ for all $y>0$.
\item [(A2)] $|\mu| + b$ is at most linear growth, and
$(\sigma^2)'$ is at most polynomial growth.
\end{itemize}
Under the above assumptions (A1)--(A2), we have unique non-negative, non-explosive strong solutions for both (Stk) and (Vol). Furthermore, $X^{x,t}_s>0$ for all $s>t$. Such a stock price model includes Heston model.

Provided that the value function $V(x,y,t)$ is smooth enough, applying dynamic programming principle and It\^o's formula, the value of the option $V(x,y,t)$ formally satisfies the  variational inequality
\beq  \label{eq:viV}
 \left\{
 \begin{array}{ll}
 \min\Big\{-\p_t V-{\cal L}_x V,\; V-(K-x)^+
 \Big\}=0,\quad(x,y,t)\in Q:=\mathbb{R}^+\times \mathbb{R}^+\times[0,T),
 \vspace{2mm} \\
 V(x,y,T)=(K-x)^+,\quad (x,y)\in\mathbb{R}^+\times \mathbb{R}^+,
 \end{array}
 \right.
 \eeq
 where
 $${\cal L}_x V=\frac{1}{2}x^2\s^2(y)\p_{xx}V+\rho x\s(y)b(y)\p_{xy}V+\frac{1}{2}b^2(y)\p_{yy}V+rx\p_xV+\mu(y)
 \p_yV-rV.$$
 On the boundary $x=0$,
 the Fichera condition on linear parabolic equation suggests us not to impose any boundary condition. On the boundary $y=0$,
 the Fichera function is
 \bee
 F=\Big[\mu(y)-\frac{1}{2}\rho \s(y)b(y)-b(y)b'(y)\Big]\Big|_{y=0}
 =\mu(0)-\lim\limits_{y\rightarrow0}b(y)b'(y).
 \eee
 So
 \begin{enumerate}
 \item [(F1)] when $\mu(0)<\lim\limits_{y\rightarrow0}b(y)b'(y)$, one has to impose the boundary condition;
\item [(F2)]  when $\mu(0)\geq\lim\limits_{y\rightarrow0}b(y)b'(y)$,
 one should not impose any boundary condition.
 \end{enumerate}
Although the problem \eqref{eq:viV} is a variational inequality instead of linear PDE, the Fichera condition in the linear PDE theory does not
directly prove the existence of solution to the problem \eqref{eq:viV}.
Throughout this paper, we study the relation between value function \eqref{eq:defV} and PDE of \eqref{eq:viV} with an appropriate boundary
data on $y = 0$  under the case (F1),  
while the case (F2) is left in the future study. 

Then, what boundary condition should be imposed on the boundary $y=0$? To proceed, we define $\nu:=\inf\{s>t: Y_s=0\}$ the first hitting time of the process $Y_s$ to the boundary $y=0$. Then for any stopping time $\tau>\nu$, we have
\beq\label{eq:nu}
e^{-r(\tau-t)}(K-X_\tau)^+\leq e^{-r(\nu-t)}(K-X_\nu)^+,
\eeq
thus
\bee
V(x,0,t)=\sup\limits_{\tau\in\mathcal{T}_{t,T}}\mathbb{E}_{x,0,t}[e^{-r(\tau-t)}(K-X_\tau)^+]\leq (K-x)^+.
\eee
On the other hand, by taking $\tau = t$,
\bee
V(x,0,t)=\sup\limits_{\tau\in\mathcal{T}_{t,T}}\mathbb{E}_{x,0,t}[e^{-r(\tau-t)}(K-X_\tau)^+]\geq (K-x)^+.
\eee
Hence
 \beq\label{eq:bc}
 V(x,0,t)=(K-x)^+=\sup\limits_{\tau\in\mathcal{T}_{t,T}}\mathbb{E}_{x,0,t}[e^{-r(\tau-t)}(K-X_\tau)^+],\quad(x,t)\in \mathbb{R}^+\times[0,T).
 \eeq
 Due to the non-linearity of the variational inequality \eqref{eq:viV},
 we may not expect the heuristic argument in the above assuming
 the enough regularity on $V$. In addition, Fichera condition on the
 boundary data is suitable only for linear second order PDE,
 see \cite{OR73}. Our objective in this paper is to justify the regularity of
 the value function in the Sobolev space, so that the value function
 \eqref{eq:defV} can be characterized as the unique solution of
 the variational inequality \eqref{eq:viV} and an additional boundary data
 \eqref{eq:bc}.

\section{Solvability on the transformed problem}
\label{sec:strongsol}
In order to obtain the existence of solution to the variational inequality \eqref{eq:viV} with boundary condition \eqref{eq:bc} in Sobolev space, we consider the existence of strong solution to the associated transformed problem of  \eqref{eq:viV} with boundary condition \eqref{eq:bc} in this section.

To proceed, we take a simple logarithm transformation
to the variational inequality.
 Let $s=\ln x,\;\theta=T-t,\;u(s,y,\theta)=V(x,y,t)$, then
 \bee
 {\cal L}_x V&=&\frac{1}{2}\s^2(y)\p_{ss}u+\rho\s (y)b(y)\p_{sy}u+\frac{1}{2}b^2(y)\p_{yy}u+\Big(r-\frac{1}{2}\s^2(y)\Big)\p_su+\mu(y)
\p_yu-ru\\&:=& {\cal L}_s u.
 \eee
 Thus $u(s,y,\theta)$ satisfies
 \beq  \label{eq:viu}
 \left\{
 \begin{array}{ll}
 \min\Big\{\p_\theta u-{\cal L}_s u,\; u-(K-e^s)^+
 \Big\}=0,\;\;\;& (s,y,\theta)\in\mathcal {Q}:=\mathbb{R}\times \mathbb{R}^+\times(0,T],
 \vspace{2mm} \\
 u(s,y,0)=(K-e^s)^+,\;\;\;&s\in \mathbb{R},\;y\in\mathbb{R}^+,
 \vspace{2mm} \\
 u(s,0,\theta)=(K-e^s)^+,\;\;\;&s\in \mathbb{R},\;\theta\in(0,T].
 \end{array}
 \right.
 \eeq

Problem \eqref{eq:viu} is a variational inequality, we apply penalty approximation techniques to show the existence of strong solution to \eqref{eq:viu}. Suppose $u_\ep(s,y,\theta)$ satisfies
 \beq  \label{eq:uep}
 \left\{
 \begin{array}{ll}
 \p_\theta u_\ep-{\cal L}_s^\ep u_\ep+\beta_\ep(u_\ep-\pi_\ep(K-e^s))
 =0,\;\;\;& (s,y,\theta)\in\mathcal {Q},
 \vspace{2mm} \\
 u_\ep(s,y,0)=\pi_\ep(K-e^s),\;\;\;&s\in \mathbb{R},\;y\in\mathbb{R}^+,
 \vspace{2mm} \\
 u_\ep(s,0,\theta)=\pi_\ep(K-e^s),\;\;\;&s\in \mathbb{R},\;\theta\in(0,T].
 \end{array}
 \right.
 \eeq
 where ${\cal L}_s^\ep u=\frac{1}{2}(\s^2(y)+\ep)\p_{ss}u+\rho(\s(y) b(y)+\ep)\p_{sy}u+\frac{1}{2}(b^2(y)+\ep)\p_{yy}u+(r-\frac{1}{2}\s^2(y)-\frac{1}{2}\ep)\p_su+\mu
 (y)\p_yu-ru$, and $\beta_\ep(\xi)$ (Fig. 1.), $\pi_\ep(\xi)$ (Fig. 2.) satisfy
 \bee
 &&\beta_{\varepsilon}(\xi)\in C^{2}(-\infty,\
 +\infty),\ \beta_{\varepsilon}(\xi)\leq 0,\;
 \beta_{\varepsilon}(0)=-2r(K+1),\;
 \beta'_{\varepsilon}(\xi)\geq 0,\ \beta''_{\varepsilon}(\xi)\leq 0.\\
 &&\lim\limits_{\varepsilon \rightarrow 0}\beta _{\varepsilon
 }(\xi)=\left\{
 \begin{array}{ll}
 0, & \xi>0,\vspace{2mm} \\
 -\infty , & \xi<0,%
 \end{array}%
 \right.\quad
 \pi _{\varepsilon }(\xi)=\left\{
 \begin{array}{ll}
 \xi, & \xi\geq \varepsilon, \vspace{2mm} \\
 \nearrow, & |\xi|\leq \varepsilon, \vspace{2mm} \\
 0, & \xi\leq -\varepsilon,%
 \end{array}%
 \right.\\
 &&\pi_{\varepsilon}(\xi)\in C^{\infty},\
 0\leq\pi'_{\varepsilon}(\xi)\leq 1,\;\pi''_\ep(\xi)\geq0,\
 \lim\limits_{\varepsilon\rightarrow 0}\pi_{\varepsilon}(\xi)=\xi^{+}.
 \eee

\begin{center}
  \begin{picture}(600,120)
    \put(50,80){\vector(1,0){130}}
    \put(110,10){\vector(0,1){95}}
    \qbezier(102,15)(115,80)(120,80)
    \put(185,80){$\xi$}
    \put(120,85){$\varepsilon$}
    \put(111,50){$-C_0$}
    \put(110,52){\circle*{2}}
    \put(120,80){\circle*{2}}
    \put(70,-5){Fig. 1. \ \ $\beta_\varepsilon (\xi)$ }

    \put(230,40){\vector(1,0){155}}
    \put(300,10){\vector(0,1){95}}
    \qbezier(270,40)(298,37)(317,58)
    \put(300,40){\line(1,1){50}}
    \put(280,40){\circle*{2}}
    \put(320,40){\circle*{2}}
    \put(265,30){$-\varepsilon$}
    \put(315,30){$\varepsilon$}
    \put(390,40){$\xi$}
    \put(270,-5){Fig. 2. \ \ $\pi_\varepsilon (\xi)$ }
  \end{picture}
\end{center}

 Since $\mathcal{Q}=(-\infty,+\infty)\times(0,+\infty)\times(0,T]$ is infinitely, we consider the truncated version of \eqref{eq:uep}, denote
 $\mathcal {Q}^N=(-N,N)\times(0,N)\times(0,T]$, let $u^N_\ep(s,y,\theta)$ satisfy
 \beq  \label{eq:unep}
 \left\{
 \begin{array}{ll}
 \p_\theta u^N_\ep-{\cal L}_s^\ep u^N_\ep+\beta_\ep(u^N_\ep-\pi_\ep(K-e^s))
 =0,\;\;\;& (s,y,\theta)\in \mathcal {Q}^N,
 \vspace{2mm} \\
 u^N_\ep(s,y,0)=\pi_\ep(K-e^s),\;& (s,y)\in(-N,N)\times(0,N),
 \vspace{2mm} \\
 u^N_\ep(s,0,\theta)=\pi_\ep(K-e^s),\;& (s,\theta)\in (-N,N)\times(0,T],
 \vspace{2mm} \\
 \p_yu^N_\ep(s,N,\theta)=0,\;& (s,\theta)\in(-N,N)\times (0,T],
 \vspace{2mm}\\
 u^N_\ep(-N,y,\theta)=\pi_\ep(K-e^{-N}),\;& (y,\theta)\in(0,N)\times (0,T],
 \vspace{2mm}\\
 \p_su^N_\ep(N,y,\theta)=0,\;& (y,\theta)\in(0,N)\times (0,T].
 \end{array}
 \right.
 \eeq

\begin{lemma}
 For any fixed $\ep,\;N>0$, there exists a unique solution $u^N_\ep\in
 W^{2,1}_p(\mathcal {Q}^N)$ to the problem (\ref{eq:unep}), and
 \be
 &&\pi_\ep(K-e^s)\leq u^N_\ep(s,y,\theta)\leq K+1,\label{eq:estunep}\\
 &&\p_\theta u^N_\ep(s,y,\theta)\geq0.\label{eq:estptunep}
 \ee

\end{lemma}

\begin{proof}
For any fixed $\ep,\;N>0$, it is not hard to show by the fixed
 point theorem that problem (\ref{eq:unep}) has a solution $u^N_\ep\in
 W^{2,1}_p(\mathcal {Q}^N)$, and
 \bee
 |u^N_\ep|_{0}\leq C(|\beta_\ep(u^N_\ep-\pi_\ep(K-e^s))|_0+|\pi_\ep(K-e^s)|_0)\leq
 C(|\beta_\ep(-K-1)|+K+1)\leq C.
 \eee
   The proof of uniqueness is a standard way
 as well.

 Since
 \bee
 &&\p_\theta \pi_\ep(K-e^s)-{\cal L}_s^\ep (\pi_\ep(K-e^s))+\beta_\ep(0)\\
 &=&-\frac{1}{2}(\s^2(y)+\ep)(\pi''_\ep(\cdot)e^{2s}-\pi'_\ep(\cdot)e^s)+\Big(r-\frac{1}{2}\s^2(y)-\frac{1}{2}\ep\Big)\pi'_\ep(\cdot)e^s+r\pi_\ep(K-e^s)+\beta_\ep(0)\\
 &\leq&r(K+\ep)+r(K+\ep)+\beta_\ep(0)\leq 0.
 \eee
 Combining with the initial and boundary conditions,  we know $\pi_\ep(K-e^s)$
 is a subsolution of \eqref{eq:unep}.
 Similarly we know $K+1$ is a supersolution of \eqref{eq:unep}.

  Next we will prove (\ref{eq:estptunep}). Set
  $u^\delta(s,y,\theta):=u^N_\ep(s,y,\theta+\delta)$, then
  $u^\delta(s,y,\theta)$ satisfies
 \bee
 \left\{
 \begin{array}{ll}
 \p_\theta u^\delta-{\cal L}_s^\ep u^\delta+\beta_\ep(u^\delta-\pi_\ep(K-e^s))
 =0,
 \vspace{2mm} \\
 u^\delta(s,y,0)=u^N_\ep(s,y,\delta)\geq\pi_\ep(K-e^s)=u^N_\ep(s,y,0),
 \vspace{2mm} \\
 u^\delta(s,0,\theta)=\pi_\ep(K-e^s),
 \vspace{2mm} \\
 \p_yu^\delta(s,N,\theta)=0,
 \vspace{2mm}\\
 u^\delta(-N,y,\theta)=\pi_\ep(K-e^{-N}),
 \vspace{2mm}\\
 \p_su^\delta(N,y,\theta)=0.
 \end{array}
 \right.
 \eee
 Applying the comparison principle, we have
 \bee
 u^\delta(s,y,\theta)\geq u^N_\ep(s,y,\theta),\quad
 (s,y,\theta)\in(-N,N)\times(0,N)\times(0,T-\delta],
 \eee
 which yields $\p_\theta u^N_\ep\geq0$.
\end{proof}

Letting $N\rightarrow+\infty$, the existence of strong solution to the penalty problem \eqref{eq:uep} with some estimates is given in the following theorem.
 \begin{lemma}
 For any fixed $\ep>0$, there exists a unique solution
 $u_\ep(s,y,\theta)\in W^{2,1}_{p,loc}(\mathcal {Q})
 \cap C^1(\overline{\mathcal {Q}})$ to
 the problem (\ref{eq:uep}) for any $1<p<+\infty$, and
 \be
 &&\pi_\ep(K-e^s)\leq u_\ep(s,y,\theta)\leq K+1,\label{eq:estuep}\\
 &&\p_\theta u_\ep(s,y,\theta)\geq0,\label{eq:estptuep}\\
 &&-e^s\leq \p_su_\ep(s,y,\theta)\leq0,\label{eq:estpsuep}\\
 &&\p_yu_\ep(s,y,\theta)\geq0.\label{eq:estpyuep}
 \ee
 \end{lemma}

 \begin{proof} For any fixed $\ep>0$, $R>0$, applying
 $W^{2,1}_p$ interior estimate with part of boundary \cite{LSU67} to the problem
 (\ref{eq:unep}) ($N>R$),  then
 \bee
 |u^N_\ep|_{W^{2,1}_{p}(\mathcal {Q}^R)}&\leq&C
 (|\beta_\ep(u^N_\ep-\pi_\ep(K-e^s))|_{L^\infty(\mathcal {Q}^R)}+|\pi_\ep(K-e^s)|_
 {W^{2,1}_p(\overline{\mathcal {Q}^R}\cap \{\theta=0\})})
 \leq C,
 \eee
 where $C$ depends on $\ep, R$ but is independent of $N$.
  Letting $N\rightarrow+\infty$, by the imbedding theorem, we know
 problem (\ref{eq:uep}) has a solution $u_\ep(s,y,\theta)\in W^{2,1}_{p,loc}(\mathcal {Q})\cap
 C^1(\overline{\mathcal {Q}})$. \eqref{eq:estuep}--\eqref{eq:estptuep} are
 consequences of \eqref{eq:estunep}--\eqref{eq:estptunep}.

 Now we aim to prove (\ref{eq:estpsuep}). Differentiate the equation in (\ref{eq:uep}) w.r.t. $s$
 and denote $w_1=\p_s
 u_\ep$, then
 \beq    \label{eq:psuep}
 \left\{
 \begin{array}{ll}
 \p_\theta w_1-{\cal L}_s^\ep
 w_1+\beta'_\ep(\cdot)w_1=-\beta'_\ep(\cdot)\pi'_\ep(\cdot)e^s
 \leq0,\;\;\;(s,y,\theta)\in\mathcal {Q},
 \vspace{2mm} \\
 w_1(s,y,0)=w_1(s,0,\theta)=-\pi'_\ep(\cdot)e^s\leq0.
 \end{array}
 \right.
 \eeq
 Applying the maximum principle \cite{Tso85} we know $w_1=\p_s
 u_\ep\leq0$. In view of
 \bee
 (\p_\theta-{\cal L}_s^\ep)(-e^s)+\beta'_\ep(\cdot)(-e^s)
 =-\beta'_\ep(\cdot)e^s\leq-\beta'_\ep(\cdot)\pi'_\ep(\cdot)e^s.
 \eee
 Combining with the initial and boundary conditions, applying the
 comparison principle we have
 \bee
 -e^s\leq w_1(s,y,\theta)=\p_s
 u_\ep(s,y,\theta)\leq0.
 \eee

 Finally we want to prove (\ref{eq:estpyuep}). We first differentiate
 (\ref{eq:psuep}) w.r.t. $s$, denote $w_2=\p_{ss}u_\ep$, we obtain
 \beq\label{eq:pssuep}
 \left\{
 \begin{array}{ll}
 \p_\theta w_2-{\cal L}_s^\ep
 w_2+\beta'_\ep(\cdot)w_2=-\beta'_\ep(\cdot)\pi'_\ep(\cdot)e^s+\beta'_\ep(\cdot)\pi''_\ep(\cdot)e^{2s}
 -\beta''_\ep(\cdot)[\pi'_\ep(\cdot)e^s+w_1]^2,
 \vspace{2mm} \\
 w_2(s,y,0)=w_2(s,0,\theta)=-\pi'_\ep(\cdot)e^s+\pi''_\ep(\cdot)e^{2s}.
 \end{array}
 \right.
 \eeq
 Set $w_3(s,y,\theta):=w_2(s,y,\theta)-w_1(s,y,\theta)$, in view of (\ref{eq:psuep}) and (\ref{eq:pssuep})
 \bee
 \left\{
 \begin{array}{ll}
 \p_\theta w_3-{\cal L}_s^\ep
 w_3+\beta'_\ep(\cdot)w_3=\beta'_\ep(\cdot)\pi''_\ep(\cdot)e^{2s}
 -\beta''_\ep(\cdot)[\pi'_\ep(\cdot)e^s+w_1]^2\geq0,
 \vspace{2mm} \\
 w_3(s,y,0)=w_3(s,0,\theta)=\pi''_\ep(\cdot)e^{2s}\geq0.
 \end{array}
 \right.
 \eee
 Applying maximum principle we know $w_3(s,y,\theta)\geq0$, i.e., $\p_{ss} u_\ep-\p_s
 u_\ep\geq0$.

 Differentiate (\ref{eq:uep}) w.r.t. $y$, denote
 $w_4(s,y,\theta)=\p_yu_\ep(s,y,\theta)$. Then we get
 \bee
 \left\{
 \begin{array}{ll}
 \p_\theta w_4-{\cal L}_s^\ep
 w_4-\rho(\s'(y)b(y)+\s (y)b'(y))\p_sw_4-b(y)b'(y)\p_yw_4\\
 \hspace{2cm}-\mu'(y)w_4+\beta'_\ep(\cdot)w_4=\s(y)\s'(y)(\p_{ss}u_\ep-\p_su_\ep),
 \vspace{2mm} \\
 w_4(s,y,0)=0,
 \vspace{2mm} \\
 w_4(s,0,\theta)\geq0.
 \end{array}
 \right.
 \eee
 Since $\s'(y)\geq0$, $u_\ep\in
 C^{2,1}(\mathcal {Q})$ and $\p_{ss} u_\ep-\p_s u_\ep\geq0$,
 by maximum principle \cite{Tso85}, we have $\p_yu_\ep(s,y,\theta)\geq0$.
 \end{proof}

Now we are able to show the solvability on the variational inequality \eqref{eq:viu} in the Sobolev space by the approximation of a subsequence of $\{u_\ep\}$.
 \begin{lemma}\label{lem:u} There exists a solution $u\in W^{2,1}_{p}(\mathcal {Q}_\delta^N\setminus B_h)
$ to the problem (\ref{eq:viu}), where $\mathcal
{Q}_\delta^N=(-N,N)\times(\delta,N)\times(0,T],\;B_h=(\ln K-h,\ln
 K+h)\times(0,+\infty)\times(0,T]$ for any $N,\;\delta,\;h>0$. Moreover,
 \be
 &&(K-e^s)^+\leq u(s,y,\theta)\leq K+1,\label{eq:estu}\\
 &&\p_\theta u(s,y,\theta)\geq0,\label{eq:estptu}\\
 &&-e^s\leq \p_su(s,y,\theta)\leq0,\label{eq:estpsu}\\
 &&\p_yu(s,y,\theta)\geq0.\label{eq:estpyu}
 \ee
 \end{lemma}
 \begin{proof} Since $\s(y),\ b(y)$ are continuous
 and $\s'(y)\geq0$,
 in $\mathcal {Q}_{\frac{1}{2}}^N$, we have
 $\s^2(y)+\ep\geq\s^2(\frac{1}{2})>0$, and $\lambda_1|\xi|^2\leq a^{ij}\xi_i\xi_j\leq
 \Lambda_1|\xi|^2$, with $\Lambda_1,\;
 \lambda_1$ independent of $\ep$.
 Applying $C^{\alpha,\alpha/2}$ estimate \cite{Lie96}
  and $W^{2,1}_p$ interior estimate with part of boundary \cite{LSU67}, we have
 \bee
 |u_\ep|_{C^{\alpha,\alpha/2}\big(\overline{\mathcal {Q}_{\frac{1}{2}}^N}\big)}&\leq&
 C(|u_\ep|_{0}+|\beta_\ep(u_\ep-\pi_\ep(K-e^s))|_{0}+[\pi_\ep(K-e^s)]_{C^\gamma(-N,N)})
 \leq C_1,\\
 |u_\ep|_{W^{2,1}_p(\mathcal {Q}_{\frac{1}{2}}^N\setminus B_h)}&\leq&
 C(|u_\ep|_{L^\infty}+|\beta_\ep(u_\ep-\pi_\ep(K-e^s))|_{L^\infty}+|(K-e^s)\vee0|_
 {W^{2,1}_p(\mathcal {Q}_{\frac{1}{2}}^N\setminus B_h)})
 \leq C_2,
 \eee
 where $C_1,\;C_2$ are independent of $\ep$ due to the estimate \eqref{eq:estuep}
 and the definitions of $\beta_\ep, \pi_\ep$.
 Thus there exists a
 subsequence of $\{u_\ep\}$, denote $\{u_\ep^{(1)}\}$, and
 $u^{(1)}\in W^{2,1}_p(\mathcal {Q}_{\frac{1}{2}}^N\setminus B_h)\cap
 C\Big(\overline{\mathcal {Q}_{\frac{1}{2}}^N}\Big)$, such that
 \bee
 &&u_\ep^{(1)}(s,y,\theta)\rightharpoonup u^{(1)}(s,y,\theta)\quad in\;W^{2,1}_p(\mathcal {Q}_{\frac{1}{2}}^N\setminus
 B_h)\;weakly,\\
 &&u_\ep^{(1)}(s,y,\theta)\rightarrow
 u^{(1)}(s,y,\theta)\quad in\;C\Big(\overline{\mathcal {Q}_{\frac{1}{2}}^N}\Big)\;uniformly.
 \eee
 In a same way, in  $\mathcal {Q}_{\frac{1}{3}}^N$, we have
 $\s^2(y)+\ep\geq\s^2(\frac{1}{3})$, and $\lambda_2|\xi|^2\leq a^{ij}\xi_i\xi_j\leq
 \Lambda_2|\xi|^2$, with $\Lambda_2,\;
 \lambda_2$ independent of $\ep$.
 Thus there exists $\{u_\ep^{(2)}\}\subseteq\{u_\ep^{(1)}\},
 \;u^{(2)}\in W^{2,1}_p(\mathcal {Q}_{\frac{1}{3}}^N\setminus B_h)\cap
 C\Big(\overline{\mathcal {Q}_{\frac{1}{3}}^N}\Big)$,  such that
 \bee
 &&u_\ep^{(2)}(s,y,\theta)\rightharpoonup u^{(2)}(s,y,\theta)\quad in\;W^{2,1}_p(\mathcal {Q}_{\frac{1}{3}}^N\setminus
 B_h)\;weakly,\\
 &&u_\ep^{(2)}(s,y,\theta)\rightarrow
 u^{(2)}(s,y,\theta)\quad in\;C\Big(\overline{\mathcal {Q}_{\frac{1}{3}}^N}\Big)\;uniformly.
 \eee
 Moreover,
 \bee
 u^{(2)}(s,y,\theta)=u^{(1)}(s,y,\theta),\quad(s,y,\theta)\in \mathcal {Q}_{\frac{1}{2}}^N.
 \eee
  Define
 $u(s,y,\theta)=u^{(k)}(s,y,\theta)$, if $(s,y,\theta)\in \mathcal {Q}^N_{\frac{1}{k+1}}$,
 abstracting diagram subsequence $\{u^{(k)}_{\ep_k}\}$, for any
 $\delta,\;h,\;N>0$, we have
 \bee
 &&u^{(k)}_{\ep_k}(s,y,\theta)\rightharpoonup u(s,y,\theta)\quad in\;W^{2,1}_p(\mathcal {Q}_{\delta}^N\setminus
 B_h)\;weakly,\\
 &&u^{(k)}_{\ep_k}(s,y,\theta)\rightarrow
 u(s,y,\theta)\quad in\;C(\overline{\mathcal {Q}_{\delta}^N})\;uniformly,
 \eee
 thus $u(s,y,\theta)\in W^{2,1}_p(\mathcal {Q}_{\delta}^N\setminus B_h)\cap
 C(\overline{\mathcal {Q}}\setminus
 \{y=0\})$ and $u(s,y,\theta)$ satisfies the variational inequality in
 \eqref{eq:viu} and the initial condition.

Next we will prove the continuity on the degenerate boundary $y=0$.
 For any $s_0\in \mathbb{R}\setminus\{\ln K\}$, then there exists $\ep_0>0$ such that
 $\pi_{\ep_0}(K-e^{s_0})=(K-e^{s_0})^+$, denote $w_0(s,y,\theta)=\pi_{\ep_0}(K-e^s)+Ay^\alpha\geq0$,
 with $0<\alpha<1$, and $A\geq 1$ to be determined, then for any $\ep<\ep_0$
 \bee
 &&\p_\theta w_0-\mathcal {L}^\ep_sw_0+\beta_\ep(w_0-\pi_\ep(K-e^s))\\
 &=&-\frac{1}{2}(\s^2(y)+\ep)\pi''_{\ep_0}(\cdot)e^{2s}
 -\frac{1}{2}(b^2(y)+\ep)A\alpha(\alpha-1)y^{\alpha-2}+r\pi'_{\ep_0}(\cdot)e^s\\
 &&-\mu(y)\alpha
 Ay^{\alpha-1}+rw_0+\beta_\ep(\pi_{\ep_0}(K-e^s)-\pi_{\ep}(K-e^s)+Ay^\alpha)\\
 &\geq&-\frac{1}{2}(\s^2(y)+\ep)\pi''_{\ep_0}(\cdot)e^{2s}
 +\frac{1}{2}(b^2(y)+\ep)A\alpha(1-\alpha)y^{\alpha-2}-\mu(y)\alpha
 Ay^{\alpha-1}+\beta_\ep(0),
 \eee
 since the negative Fichera function indicates $\mu(0)<\lim\limits_{y\rightarrow0}b(y)b'(y)$, in addition,
 $b^2(y)=O(y)$ or $b^2(y)=o(y)$ when $y\rightarrow0$, hence
 there exists $\delta_0>0$ small enough and independent of
 $\ep$ such that
 \bee
 -\frac{1}{2}(\s^2(y)+\ep)\pi''_{\ep_0}(\cdot)e^{2s}
 +\frac{1}{2}(b^2(y)+\ep)A\alpha(1-\alpha)y^{\alpha-2}-\mu(y)\alpha
 Ay^{\alpha-1}+\beta_\ep(0)\geq0
 \eee
 for any $y\in (0,\delta_0)$. Moreover, we can choose $A$ large enough such that
  $A\delta_0^\alpha\geq
 K+1$. Combining
 \bee
 \pi_{\ep_0}(K-e^s)+Ay^\alpha\geq\pi_\ep(K-e^s),\quad \ep<\ep_0,
 \eee
 applying comparison principle, we have
 \bee
 \pi_\ep(K-e^s)\leq u_\ep(s,y,\theta)\leq
 \pi_{\ep_0}(K-e^s)+Ay^\alpha,\quad (s,y,\theta)\in\mathbb{R}\times(0,\delta_0)\times(0,T].
 \eee
 Letting $\ep\rightarrow0^+$ we have
 \bee
 (K-e^s)^+\leq
 u(s,y,\theta)\leq\pi_{\ep_0}(K-e^s)+Ay^\alpha,\quad (s,y,\theta)\in\mathbb{R}\times(0,\delta_0)\times(0,T].
 \eee
 In particular
 \bee
 (K-e^{s_0})^+\leq
 u(s_0,y,\theta)\leq(K-e^{s_0})^++Ay^\alpha,\;\;\;y\in(0,\delta_0).
 \eee
 Letting $y\rightarrow0^+$, we obtain
 \bee
 u(s_0,0,\theta)=(K-e^{s_0})^+,\;\;\;\theta\in(0,T],
 \eee
 since $s_0$ is arbitrary, then $u(s,0,\theta)=(K-e^s)^+,\;s\in\mathbb{R}\setminus\{\ln
 K\}$.
 Therefore $u(s,y,\theta)$ is a solution to the problem (\ref{eq:viu}).
(\ref{eq:estu})--(\ref{eq:estpyu}) are consequences of
(\ref{eq:estuep})--(\ref{eq:estpyuep}).
 \end{proof}

\section{Characterization of free boundary to the problem \eqref{eq:viu}}\label{sec:freeboundary}
Variational inequality \eqref{eq:viu} is an obstacle problem, this section aims to characterize the free boundary arise from \eqref{eq:viu}.

 \begin{lemma}\label{thm:estu}
  The solution to the problem (\ref{eq:viu})
 satisfies
 $$u(s,y,\theta)>0,\quad (s,y,\theta) \in \mathbb{R}\times\mathbb{R}^+\times(0,T].$$
 \end{lemma}

 \begin{proof} For any
 fixed $y_0>0$, we have
 \bee
 \left\{
 \begin{array}{ll}
 \p_\theta u-{\cal L}_su\geq
 0,\;\;\;& (s,y,\theta)\in\mathbb{R}\times(y_0,+\infty)\times(0,T],
 \vspace{2mm} \\
 u(s,y,0)=(K-e^s)^+\geq0,& (s,y)\in\mathbb{R}\times(y_0,+\infty),
 \vspace{2mm} \\
 u(s,y_0,\theta)\geq(K-e^s)^+\geq0,& (s,\theta)\in\mathbb{R}\times(0,T].
 \end{array}
 \right.
 \eee
 Applying strong maximum principle, we obtain
 \bee
 u(s,y,\theta)>0,\quad (s,y,\theta) \in
 \mathbb{R}\times(y_0,+\infty)\times(0,T].
 \eee
 Since $y_0$ is arbitrary, then we know
 \bee
 u(s,y,\theta)>0,\quad(s,y,\theta) \in
 \mathbb{R}\times\mathbb{R}^+\times(0,T].
 \eee
 \end{proof}

 In order to characterize the free boundary, we first define
 \bee
 &&\mathcal {C}[u]:=\{(s,y,\theta): u(s,y,\theta)=(K-e^s)^+\}{\rm(Coincidence\ set)},\\
 &&\mathcal {N}[u]:=\{(s,y,\theta): u(s,y,\theta)>(K-e^s)^+\}{\rm(Noncoincidence\ set)}.
 \eee
Thanks to the estimates \eqref{eq:estu}--\eqref{eq:estpyu} of the solution to \eqref{eq:viu}, problem \eqref{eq:viu} gives rise to a free boundary that can be expressed as a function of  $(y,\theta)$. The following three lemmas give the existence and properties of the free boundary.
 \begin{proposition}\label{prop:h} There exists
 $h(y,\theta): \mathbb{R}^+\times(0,T]\rightarrow\mathbb{R}$,
 such that
  \beq\label{eq:C}
 \mathcal {C}[u]=\{(s,y,\theta)\in \mathcal {Q}: s\leq
 h(y,\theta),\;y\in\mathbb{R}^+,\;\theta\in(0,T]\}.
 \eeq
 Moreover, for any fixed $y>0,\;h(y,\theta)$ is monotonic decreasing w.r.t.
 $\theta$; for any fixed $\theta\in(0,T],\;h(y,\theta)$ is monotonic decreasing w.r.t.
 $y$.
 \end{proposition}
 \begin{proof}
 Since $(K-e^s)^+=0$ when $s\geq\ln K$, in view of Lemma
 \ref{thm:estu}, we have
 \bee
 \{s\geq\ln K\}\subset \mathcal {N}[u],\quad\mathcal {C}[u]\subset \{s<\ln K\}.
 \eee
Hence problem  (\ref{eq:viu}) is equivalent to the following problem
 \bee
 \left\{
 \begin{array}{ll}
 \min\Big\{\p_\theta u-{\cal L}_s u,\; u-(K-e^s)
 \Big\}=0,\;\;\;& (s,y,\theta)\in\mathcal {Q}:=\mathbb{R}\times \mathbb{R}^+\times(0,T],
 \vspace{2mm} \\
 u(s,y,0)=(K-e^s)^+,\;\;\;&s\in \mathbb{R},\;y\in\mathbb{R}^+,
 \vspace{2mm} \\
 u(s,0,\theta)=(K-e^s)^+,\;\;\;&s\in \mathbb{R},\;\theta\in(0,T].
 \end{array}
 \right.
 \eee
Together with (\ref{eq:estpsu}),
  we can define
 \bee
 h(y,\theta):=\max\{s\in\mathbb{R}: u(s,y,\theta)=(K-e^s)\},\quad(y,\theta)\in\mathbb{R}^+\times(0,T],
 \eee
 by the definition of $h(y,\theta)$, we know \eqref{eq:C} is true.

 Suppose
 $h(y,\theta_1)=s_1$, notice that $\p_\theta u(s,y,\theta)\geq0$,  then for any $\theta_2\leq \theta_1$,
 \bee
 0\leq u(s_1,y,\theta_2)-(K-e^{s_1})\leq
 u(s_1,y,\theta_1)-(K-e^{s_1})=0,
 \eee
 from which we infer that
 \bee
 u(s_1,y,\theta_2)=(K-e^{s_1}),\quad\theta_2\leq \theta_1.
 \eee
  By
 the definition of $h(y,\theta)$, we know $h(y,\theta_2)\geq s_1=h(y,\theta_1)$, thus
 $h(y,\cdot)$ is monotonic decreasing w.r.t. $\theta$.

 Similarly, the monotonicity of $h(y,\theta)$ w.r.t. $y$ can be
 deduced by virtue of $\p_yu(s,y,\theta)\geq0$ and the definition of $h(y,\theta)$.
 \end{proof}

 \begin{proposition}\label{prop:h1}  $h(y,\theta)$ is continuous  on $\mathbb{R}^+\times[0,T]$ with
 \bee
 h(y,0):=
 \lim\limits_{\theta\rightarrow0^+}h(y,\theta)=\ln K,\quad y>0.
 \eee
 \end{proposition}

 \begin{proof} We first prove $h(y,\theta)$ is continuous w.r.t.
 $\theta$. Suppose not. There exists $y_0>0,\;\theta_0>0$ such that
 $s_1:=h(y_0,\theta_0+)<h(y_0,\theta_0):=s_2$. Since $h(y_0,\theta_0+)=s_1$(see Fig. 3.), then
 \bee
u(s,y_0,\theta)>K-e^s,\quad s>s_1,\ \theta>\theta_0.
\eee
In fact,
$u\in W^{2,1}_p$ and the embedding theorem imply that $u$ is
uniformly continuous, thus there exists $\delta>0$, take $\mathcal
S_0=(s_1, s_2)\times(y_0-\delta, y_0)$ such that $U_0:=\mathcal
S_0\times(\theta_0, T]\subseteq \mathcal N[u]$, then
 \bee
 \p_\theta u-\mathcal {L}_su=0,\quad (s,y,\theta)\in U_0.
 \eee
 Moreover, in view of $h(y_0,\theta_0):=s_2$,
then
\bee
h(y,\theta_0)\geq s_2,\quad y_0-\delta<y\leq y_0,
\eee
hence
\bee
u(s,y,\theta_0)=K-e^s,\quad s\leq s_2,\  y_0-\delta<y\leq y_0. \eee In
particular, \bee
 u(s,y,\theta_0)=K-e^s,\quad(s,y)\in\overline{U_0}\cap\{\theta=\theta_0\},
\eee
 thus
 \bee
 \p_\theta u\big|_{\theta=\theta_0}&=&\mathcal
 {L}_su\big|_{\theta=\theta_0}=\frac{1}{2}\s^2(y)(-e^s)+\Big(r-\frac{1}{2}\s^2(y)\Big)(-e^s)-r(K-e^s)\\
 &=&-rK<0,
 \eee
 which comes to a contradiction with the fact that $\p_\theta u\geq0$.
 Hence $h(y,\theta)$ is continuous w.r.t. $\theta$.

 Since $h(y,\theta)$ is monotonic decreasing w.r.t. $\theta$, then we
 can define $h(y,0):=
 \lim\limits_{\theta\rightarrow0^+}h(y,\theta)$.
 In the same way we can prove $h(y,0)=\ln K$.

 Now we aim to prove the continuity of $h(y,\theta)$ w.r.t. $y$.
 If this is not true, there exists
 $\theta_0,\;y_0>0$ such that
$ s_1:= h(y_0+,\theta_0)<h(y_0,\theta_0):=
 s_2$.
 Since $h(y,\theta)$ is continuous w.r.t. $\theta$ and $u(s,y,\theta)\in C^{1,1,1/2}(\mathcal {Q})$,
take $U_0:=(\tilde{s}, \bar{s})\times(y_0,+\infty)\times(\tilde{\theta}, \bar{\theta})$, where
$(\tilde{s}, \bar{s})\times(\tilde{\theta}, \bar{\theta})\subseteq (h(y_0+,\theta), h(y_0,\theta))$(see Fig. 4.), then
$u(s,y,\theta)$ satisfies
 \bee
 &&\p_\theta u-\mathcal {L}_su=0,\quad (s,y,\theta)\in U_0,\\
 &&u(s,y_0,\theta)=K-e^s,\quad (s,\theta)\in\overline{U_0}\cap\{y=y_0\},\\
 &&\p_yu(s,y_0,\theta)=0,\quad (s,\theta)\in\overline{U_0}\cap\{y=y_0\}.
 \eee
 Thus
 \bee
 \p_\theta u(s,y_0,\theta)=\p_{y\theta}u(s,y_0,\theta)=0,\quad (s,\theta)\in\overline{U_0}\cap\{y=y_0\}.
 \eee
 Since $\p_\theta u\geq0$ and $\p_\theta(\p_\theta u)-\mathcal {L}_s(\p_\theta u)=0$ in
 $U_0$, by Hopf lemma we know
 \bee
 \p_{y\theta}u(s,y_0,\theta)>0,\quad (s,\theta)\in\overline{U_0}\cap\{y=y_0\},
 \eee
 or
 \bee
 \p_\theta u(s,y,\theta)\equiv0,\quad (s,y,\theta)\in U_0,
 \eee
 but both come to contradictions.

Together with the monotonicity of $h(y,\theta)$ w.r.t. $y$ and $\theta$, we conclude that $h(y,\theta)$ is continuous on $\mathbb{R}^+\times[0,T].$
\end{proof}

\begin{center}
 \begin{picture}(600,120)
 \thinlines
\put(50,20){\vector(1,0){170}}
 \put(50,20){\vector(0,1){100}}
 \qbezier(115,70)(83,90)(81,115)
 \qbezier(135,70)(140,40)(176,20)
\put(135,70){\circle*{2}}
 \multiput(50,70)(5,0){17}{\line(1,0){3}}
 \multiput(115,20)(0,5){10}{\line(0,1){3}}
 \multiput(135,20)(0,5){10}{\line(0,1){3}}
 \put(110,10){$s_1$}
 \put(130,10){$s_2$}
 \put(38,66){$\theta_0$}
 \put(223,18){$s$}
 \put(42,115){$\theta$}
 \put(35,-5){Fig. 3. Discontinuity of $h(y,\theta)$ w.r.t. $\theta$}
 \put(250,20){\vector(1,0){170}}
 \put(250,20){\vector(0,1){100}}
 \qbezier(385,20)(305,35)(280,100)
 \qbezier(385,20)(335,30)(315,100)
 \put(318,90){$h(y_0,\theta)$}
 \put(260,55){$h(y_0+,\theta)$}
 \put(423,18){$s$}
 \put(243,115){$\theta$}
 \multiput(290,20)(0,5){12}{\line(0,1){3}}
 \put(290,20){\circle*{2}}
 \put(286,12){$s_1$}
 \multiput(250,78)(5,0){15}{\line(1,0){3}}
 \put(250,78){\circle*{2}}
 \put(238,75){$\theta_0$}
 \multiput(322,20)(0,5){12}{\line(0,1){3}}
 \put(322,20){\circle*{2}}
 \put(320,12){$s_2$}
 \put(385,20){\circle*{2}}
 \put(372,10){$\ln K$}
 \put(240,-5){Fig. 4.  Discontinuity of $h(y,\theta)$ w.r.t. $y$}
 \end{picture}
 \end{center}

 \begin{proposition}\label{prop:h2} The free boundary $h(y,\theta)$ satisfies
 \bee
 h_0(y)\leq h(y,\theta)< \ln K,\quad y>0,\;\theta\in(0,T],
 \eee
 where $h_0(y)$ is the free boundary curve of
 \beq \label{eq:viuinfty}
 \min\{-{\cal L}_su_\infty(s,y),\;u_\infty(s,y)-(K-e^s)^+\}=0
 ,\quad (s,y)\in\mathbb{R}\times\mathbb{R}^+.
 \eeq
 \end{proposition}
 \begin{proof} Since $\p_\theta u_\infty(s,y)=0$, then we can
 rewrite (\ref{eq:viuinfty}) as
 \begin{eqnarray*}
 \left\{
 \begin{array}{ll}
 \min\Big\{\p_\theta u_\infty-{\cal L}_s u_\infty,\; u_\infty-(K-e^s)^+
 \Big\}=0,\quad (s,y,\theta)\in\mathcal {Q},
 \vspace{2mm} \\
 u_\infty(s,y)|_{\theta=0}=u_\infty(s,y)\geq(K-e^s)^+=u(s,y,0).
 \end{array}
 \right.
 \end{eqnarray*}
 Applying the monotonicity of solution of variational inequality
 w.r.t. initial value, we have
 \bee
 u_\infty(s,y)\geq
 u(s,y,\theta),\quad\theta\geq0.
 \eee
 By the definitions of $h_0(y)$ and $h(y,\theta)$, we know
 \bee
 h_0(y)\leq h(y,\theta).
 \eee

Now we will prove $h(y,\theta)<\ln K$. Suppose not. There exists $y_0>0$, $\theta_0>0$ such that $h(y_0,\theta_0)=\ln K$. Then by the monotonicity of $h(y,\theta)$ and the fact that $h(y,\theta)\leq \ln K$, we have
\bee
u(s,y,\theta)=K-e^s,\quad (s,y,\theta)\in[0,\ln K]\times(0,y_0]\times(0,\theta_0].
\eee
Thus
\bee
\p_\theta u(\ln K,y,\theta)=\p_{s\theta} u(\ln K,y,\theta)=0,\quad (y,\theta)\in (0,y_0]\times(0,\theta_0].
\eee
Since $\p_\theta u(s,y,\theta)\geq0$ and $\p_\theta(\p_\theta u)-\mathcal{L}_s(\p_\theta u)=0, s>\ln K$. By Hopf lemma \cite{Fri64} we know
\bee
\p_{s\theta}u(\ln K,y,\theta)>0,\quad (y,\theta)\in (0,y_0]\times(0,\theta_0],
\eee
or
\bee
\p_\theta u(s,y,\theta)=0, \quad (s,y,\theta)\in(\ln K,+\infty)\times(0,y_0]\times(0,\theta_0],
\eee
but both come to contradictions.
 \end{proof}

\begin{remark}

The numerical result of $h_0(y)$ under Heston model is given in \cite{ZC11},
and by similar methods, under the assumptions (A1)--(A2), we can also obtain the existence of $h_0(y)$.
\end{remark}

\section{Characterization of the value function}\label{sec:uniqueness}

Now, we are ready to present the characterization of the value function
of \eqref{eq:defV} to the variational inequality \eqref{eq:viV} with
boundary condition \eqref{eq:bc}. To proceed, we present the solvability and
regularity results on the variational inequality \eqref{eq:viV} with
boundary condition \eqref{eq:bc}  via the counterpart on the transformed problem obtained in the above.

\begin{theorem}\label{th:V} Suppose a bounded function $v(x,y,t)\in W^{2,1}_{p,loc}(Q)\cap C(\bar{Q})$ satisfies variational inequality \eqref{eq:viV} with boundary condition \eqref{eq:bc}, the following assertions hold.
\begin{enumerate}
\item $v(x,y,t)$ satisfies the following estimates
\be
&&(K-x)^+\leq v(x,y,t)\leq K+1,\label{A0}\\
&&-1\leq \p_xv(x,y,t)\leq 0,\label{A1}\\
&&\p_yv(x,y,t)\geq0.\label{A2}
\ee
\item There exists a continuous function $g(y,t):\mathbb{R}^+\times[0,T)\rightarrow\mathbb{R}^+$, such that for any fixed $y>0$, $g(y,t)$ is monotonic increasing w.r.t. $t$; for any fixed $t\in[0,T)$, $g(y,t)$ is monotonic decreasing w.r.t. $y$ with
\bee
g(y,t)<g(y,T)=K,\quad y>0,\ t\in[0,T),
\eee
and
\begin{eqnarray*}
 \left\{
 \begin{array}{ll}
 -\p_t v(x,y,t)-{\cal L}_x v(x,y,t)=0,&\quad x>g(y,t),
 \vspace{2mm} \\
 v(x,y,t)=(K-x)^+,& 0\leq x\leq g(y,t).
 \end{array}
 \right.
 \end{eqnarray*}
\item Especially, $v(x,y,t)\in C^{2,1}$ when $x>g(y,t)$.
\end{enumerate}
\end{theorem}
\begin{proof}
By the transformations $s=\ln x, \theta=T-t, u(s,y,\theta)=v(x,y,t)$, noting $x\p_xv(x,y,t)=\p_su(s,y,\theta)$ and using estimates \eqref{eq:estu}, \eqref{eq:estpsu} and \eqref{eq:estpyu}, we can obtain \eqref{A0}--\eqref{A2}.

Let $g(y,t)=\exp\{h(y,\theta)\}$, by proposition \ref{prop:h}--\ref{prop:h1}, we can conclude 2.

For  any $x_0>g(y_0,t_0)$, then $v(x_0,y_0,t_0)>(K-x_0)^+$, since $v(x,y,t)$ is uniformly continuous, there exists a disk $B_\de(x_0,y_0,t_0)$ with center $(x_0,y_0,t_0)$ and radius $\de$ such that
\bee
v(x,y,t)>(K-x)^+,\quad (x,y,t)\in B_\de(x_0,y_0,t_0).
\eee
Applying $C^{2,1}$ interior estimate to
\bee
\p_tv(x,y,t)+\mathcal{L}_xv(x,y,t)=0,\quad (x,y,t)\in B_\de(x_0,y_0,t_0),
\eee
to obtain $v(x,y,t)\in C^{2,1}(B_\de(x_0,y_0,t_0))$, hence $v(x,y,t)\in C^{2,1}$ when $x>g(y,t)$.
\end{proof}

Finally, the uniqueness result is given in this below through the
arguments of verification theorem.
 \begin{theorem} Suppose there exists $v(x,y,t)\in W^{2,1}_{p,loc}(Q)$ to the problem \eqref{eq:viV} with boundary condition \eqref{eq:bc}, then $v(x,y,t)\geq V(x,y,t)$. If, in addition, there exists the  region
 $\mathcal{N}[v]:=\{(x,y,t)\in Q, v(x,y,t)>(K-x)^+\}$ satisfies
\bee
(\p_tv+\mathcal{L}_xv)(X_s,Y_s,s)=0,\quad s\in[t,\tau^*],
\eee
for the stopping time
$\tau^*:=\inf\{s>t: (X_s,Y_s,s)\notin\mathcal{N}[v]\}\wedge T$. Then the variational inequality \eqref{eq:viV} with boundary condition \eqref{eq:bc} admits a unique solution in $W^{2,1}_{p,loc}(Q)$ and $v(x,y,t)=V(x,y,t)$.
 \end{theorem}
 \begin{proof}
 Let $\tau^\beta_x:=\inf\{s>t: X_s\leq \frac{1}{\beta}\ \rm{or}\  X_s\geq\beta\}\wedge T$ be the first hitting time of the process $X_s$ to the upper bound $\beta$ or the lower bound $\frac{1}{\beta}$ or terminal time $T$, $\tau^\beta_y:=\inf\{s>t:Y_s\leq \frac{1}{\beta}\ \rm{or}\   Y_s\geq\beta\}\wedge T$ be the first hitting time of the process $Y_s$ to the upper bound $\beta$ or the lower bound $\frac{1}{\beta}$ or terminal time $T$. Let $\tau\in\mathcal{T}_{t,\tau^\beta_x\wedge\tau^\beta_y}$, by the general It\^o's formula \cite{Kry80},
\be\label{eq:veri}
e^{-r(\tau-t)}v(X_\tau,Y_\tau,\tau)&=&v(x,y,t)+\int_t^\tau e^{-r(s-t)}(\p_tv+\mathcal{L}_xv)(X_s,Y_s,s)ds\nonumber\\
&&+\int_t^\tau e^{-r(s-t)}[\s(Y_s)X_s\p_xvdW_s+b(Y_s)\p_yvdB_s].
\ee
Since $v(x,y,t)$ is bounded and the It\^o integrals in \eqref{eq:veri} are local martingales, hence they are martingales. Moreover, by Theorem \ref{th:V}, we know $v(x,y,t)$ satisfies $\p_tv+\mathcal{L}_xv\leq0$, $v(X_\tau,Y_\tau,\tau)\geq (K-X_\tau)^+$, hence
\bee
v(x,y,t)\geq \mathbb{E}_{x,y,t}[e^{-r(\tau-t)}(K-X_\tau)^+],\quad \tau\in \mathcal{T}_{t,\tau^\beta_x\wedge\tau^\beta_y}.
\eee
Since $X_s$ is a positive, non-explosive local martingale, then
\beq\label{eq:taux}
\lim\limits_{\beta\rightarrow \infty}\tau^\beta_x=T,\quad a.s.-\mathbb{P}.
\eeq
Since $Y_s$ is non-negative, non-explosive local martingale, then
\beq\label{eq:tauy}
\lim\limits_{\beta\rightarrow \infty}\tau^\beta_y=\nu\wedge T,\quad a.s.-\mathbb{P},
\eeq
where $\nu$ is the first hitting time of $Y_s$ to the boundary $y=0$.
Hence the arbitrariness of $\tau\in\mathcal{T}_{t,\tau^\beta_x\wedge\tau^\beta_y}$ and the above two limits imply that
\beq\label{eq:EnuT}
v(x,y,t)\geq \sup\limits_{\tau\in \mathcal{T}_{t,\nu\wedge T}}\mathbb{E}_{x,y,t}[e^{-r(\tau-t)}(K-X_\tau)^+].
\eeq
In view of \eqref{eq:nu}, when $\nu<T$,
\bee
\mathbb{E}_{x,y,t}[e^{-r(\nu-t)}(K-X_\nu)^+]\geq\mathbb{E}_{x,y,t}[e^{-r(\tau-t)}(K-X_\tau)^+],\quad\tau\in\mathcal{T}_{\nu,T}.
\eee
Together with \eqref{eq:EnuT}, we have
\bee
v(x,y,t)\geq \sup\limits_{ \tau\in \mathcal{T}_{t,T}}\mathbb{E}_{x,y,t}[e^{-r(\tau-t)}(K-X_\tau)^+]=V(x,y,t).
\eee

On the other hand, define $\tilde\tau^\beta_x:=\inf\{s>t:  X_s\geq\beta\}\wedge T$ be the first hitting time of the process $X_s$ to the upper bound $\beta$ or terminal time $T$, $\tilde\tau^\beta_y:=\inf\{s>t: Y_s\geq\beta\}\wedge T$ be the first hitting time of the process $Y_s$ to the upper bound $\beta$ or terminal time $T$. By \eqref{eq:taux} and \eqref{eq:tauy} we know
\bee
\lim\limits_{\beta\rightarrow \infty}\tilde\tau^\beta_x\wedge\tilde\tau^\beta_y=T,\quad a.s.-\mathbb{P}.
\eee
Together with Monotone Convergence Theorem, we  obtain
 \be
 && \lim_{\beta\to \infty} \mathbb{E}_{x,t} \Big[X_T
      I_{\{\tilde\tau^\beta_x= T\}} \Big]  = \mathbb{E}_{x,t} \Big[
      \lim_{\beta \to \infty} X_T
      I_{\{\tilde\tau^\beta_x = T\}} \Big] = \mathbb{E}_{x,t} \big[
      X_T\big],\label{eq:MCTX}\\
 && \lim_{\beta\to \infty} \mathbb{E}_{y,t} \Big[Y_T
      I_{\{\tilde\tau^\beta_y= T\}} \Big]  = \mathbb{E}_{y,t} \Big[
      \lim_{\beta \to \infty} Y_T
      I_{\{\tilde\tau^\beta_y = T\}} \Big] = \mathbb{E}_{y,t} \big[
      Y_T\big].\label{eq:MCTY}
      \ee
Moreover, by the definitions of $\tilde\tau^\beta_x,  \tilde\tau^\beta_y$, we can have
 \bee
\mathbb E_{x,t}\big[X_{\tilde\tau^\beta_x}\big] &=& \mathbb E_{x,t}
\Big[X_{\tilde\tau^\beta_x} I_{\{\tilde\tau^\beta_x<T\}}\Big] + \mathbb
E_{x,t}\Big[X_T I_{\{\tilde\tau^\beta_x=T\}}\Big]\\& = &\beta
\mathbb{P}\{\tilde\tau^{\beta}_x < T\}
  + \mathbb{E}_{x,t} \Big[X_T I_{\{\tilde\tau^\beta_x=T\}}\Big].
 \eee
Forcing the limit $\beta\to \infty$, due to \eqref{eq:MCTX},
  $$\lim_{\beta \to \infty} \mathbb{E}_{x,t}[X_{\tilde\tau^\beta_x}] =
  \lim_{\beta \to \infty} \beta \mathbb{P}\{\tilde\tau^{\beta}_x < T\}
  + \mathbb{E}_{x,t} [X_T].$$
  For all $\beta>x$, since $\{X_{\tilde\tau^\beta_x\wedge s}:s>t\}$ is a
  bounded local martingale, hence it is a martingale. So,
  $\mathbb{E}_{x,t}[X_{\tilde\tau^\beta_x}] = x$ for all $\beta>x$. Rearranging
  the above equality,   we have
  \beq\label{eq:betaX}
  \lim_{\beta \to \infty} \beta \mathbb{P}\{\tilde\tau^{\beta}_x < T\}
    = x -  \mathbb{E}_{x,t} [X_T]\leq x.
  \eeq
Similarly,
 \beq\label{eq:betaY}
  \lim_{\beta \to \infty} \beta \mathbb{P}\{\tilde\tau^{\beta}_y < T\}
    = y -  \mathbb{E}_{y,t} [Y_T]\leq y.
  \eeq

By Theorem \ref{th:V} we know $\mathcal{N}[v]=\{(x,y,t)\in Q: x>g(y,t)\}$, noting that $v(x,y,t)\in C^{2,1}$ and $\p_tv+\mathcal{L}_xv=0$ in $\mathcal{N}[v]$, using the classical It\^o's formula \cite{Ok03} in $[t,\tau^*\wedge\tilde\tau^\beta_x\wedge\tilde\tau^\beta_y]$, we have
\bee
v(x,y,t)&=&\mathbb{E}_{x,y,t}\Big[e^{-r(\tau^*\wedge\tilde\tau^\beta_x\wedge\tilde\tau^\beta_y- t)}v(X_{\tau^*\wedge\tilde\tau^\beta_x\wedge\tilde\tau^\beta_y},
Y_{\tau^*\wedge\tilde\tau^\beta_x\wedge\tilde\tau^\beta_y},\tau^*\wedge\tilde\tau^\beta_x\wedge\tilde\tau^\beta_y)\Big]\\
&=&\mathbb{E}_{x,y,t}\Big[e^{-r(\tau^*-t)}v(X_{\tau^*},Y_{\tau^*},\tau^*)I_{\{\tau^*\leq\tilde\tau^\beta_x\wedge\tilde\tau^\beta_y\}}\Big]\\
&&+\mathbb{E}_{x,y,t}\Big[e^{-r(\tilde\tau^\beta_x\wedge\tilde\tau^\beta_y-t)}v(X_{\tilde\tau^\beta_x\wedge\tilde\tau^\beta_y},Y_{\tilde\tau^\beta_x\wedge\tilde\tau^\beta_y},\tilde\tau^\beta_x\wedge\tilde\tau^\beta_y)I_{\{\tau^*>\tilde\tau^\beta_x\wedge\tilde\tau^\beta_y\}}\Big].
\eee
Forcing $\beta\rightarrow+\infty$, since $\lim\limits_{\beta\rightarrow \infty}\tilde\tau^\beta_x\wedge\tilde\tau^\beta_y=T$,
\bee
v(x,y,t)&=&\mathbb{E}_{x,y,t}[e^{-r(\tau^*-t)}(K-X_{\tau^*})^+]\\
&&+\lim\limits_{\beta\rightarrow \infty}\mathbb{E}_{x,y,t}\Big[e^{-r(\tilde\tau^\beta_x\wedge\tilde\tau^\beta_y-t)}v(X_{\tilde\tau^\beta_x\wedge\tilde\tau^\beta_y},Y_{\tilde\tau^\beta_x\wedge\tilde\tau^\beta_y},\tilde\tau^\beta_x\wedge\tilde\tau^\beta_y)I_{\{\tau^*>\tilde\tau^\beta_x\wedge\tilde\tau^\beta_y\}}\Big].
\eee
In the following we show the limit $\lim\limits_{\beta\rightarrow \infty}\mathbb{E}_{x,y,t}\big[e^{-r(\tilde\tau^\beta_x\wedge\tilde\tau^\beta_y-t)}v(X_{\tilde\tau^\beta_x\wedge\tilde\tau^\beta_y},Y_{\tilde\tau^\beta_x\wedge\tilde\tau^\beta_y},\tilde\tau^\beta_x\wedge\tilde\tau^\beta_y)I_{\{\tau^*>\tilde\tau^\beta_x\wedge\tilde\tau^\beta_y\}}\big]$    in the above equality is 0. Since $v(x,y,t)\leq K+1$, then there exists $g(\beta)=o(\beta),\;\beta\rightarrow+\infty$ such that $v(X_{\tilde\tau^\beta_x\wedge\tilde\tau^\beta_y},Y_{\tilde\tau^\beta_x\wedge\tilde\tau^\beta_y},\tilde\tau^\beta_x\wedge\tilde\tau^\beta_y)\leq g(\beta)=o(\beta)$, together with \eqref{eq:betaX} and \eqref{eq:betaY}, we can obtain
\bee
0&\leq& \lim\limits_{\beta\rightarrow \infty}\mathbb{E}_{x,y,t}\Big[e^{-r(\tilde\tau^\beta_x\wedge\tilde\tau^\beta_y-t)}v(X_{\tilde\tau^\beta_x\wedge\tilde\tau^\beta_y},Y_{\tilde\tau^\beta_x\wedge\tilde\tau^\beta_y},\tilde\tau^\beta_x\wedge\tilde\tau^\beta_y)I_{\{\tau^*
>\tilde\tau^\beta_x\wedge\tilde\tau^\beta_y\}}\Big]\\
&\leq&\lim\limits_{\beta\rightarrow \infty}g(\beta)\mathbb{E}_{x,y,t}\Big[e^{-r(\tilde\tau^\beta_x\wedge\tilde\tau^\beta_y-t)}I_{\{\tau^*>\tilde\tau^\beta_x\wedge\tilde\tau^\beta_y\}}\Big]\\
&\leq&\lim\limits_{\beta\rightarrow \infty}\frac{g(\beta)}{\beta}\lim\limits_{\beta\rightarrow \infty}\beta\mathbb{P}\{\tilde\tau^\beta_x\wedge\tilde\tau^\beta_y<\tau^*\}\\
&\leq&\lim\limits_{\beta\rightarrow \infty}\frac{g(\beta)}{\beta}\lim\limits_{\beta\rightarrow \infty}\beta\mathbb{P}\{\tilde\tau^\beta_x\wedge\tilde\tau^\beta_y<T\}=0.
\eee
Hence $v(x,y,t)=\mathbb{E}_{x,y,t}[e^{-r(\tau^*-t)}(K-X_{\tau^*})^+]$, therefore $v(x,y,t)=V(x,y,t)$.
 \end{proof}

\section{Conclusion}\label{sec:conclusion}
In this paper, we consider an American put option of stochastic volatility with negative Fichera function on the degenerate boundary $y=0$, we impose a proper boundary condition from the definition of the option pricing to show that the solution to the associated variational inequality is unique, which is the value of the option, and the free boundary is the optimal exercise boundary of the option. Although the asset-price volatility coefficient may grow faster than linearly and the domain is unbounded, we are able to show the uniqueness by verification theorem. In this paper we only consider the payoff function $(K-x)^+$, but the method in this paper will be useful for any  nonnegative, continuous payoff function $f(x)$ which is of strictly sublinear growth, i.e., $\lim\limits_{x\rightarrow+\infty}\frac{f(x)}{x}=0$.

\bibliographystyle{plain}

\begin{thebibliography}{10}

\bibitem{AGG10}
Aitsahlia, F.,  Goswami, M., and Guha, S.  2010.
\newblock American Option Pricing Under Stochastic Volatility: An Efficient Numerical Approach.
\newblock {\em  Computational Management Science.} 

\bibitem{BKX12}
Bayraktar, E., Kardaras, K., and Xing, H. 2012.
\newblock Valuation equations for stochastic volatility models.
\newblock {\em SIAM Journal on Financial Mathematics}, 3:351--373.

\bibitem{CSYY13}
Chen, X., Song, Q., Yi,  F., and Yin, G. 2013.
\newblock Characterization of stochastic control with optimal stopping in a Sobolev space.
\newblock{\em Automatica}, 49:1654-1662.

\bibitem{CKMZ09}
Chiarella, C., Kang, B., Meyer, G. H., and Ziogas, A. 2009.
\newblock The evaluation of American option prices under stochastic volatility and jump-diffusion dynamics using the method of lines.
\newblock {\em International Journal of Theoretical  Applied Finance}, 12: 393--425.

\bibitem{CIL92}
Crandall, M., Ishii, H., and Lions, P. 1992.
\newblock User's guide to viscosity solutions of second order partial
  differential equations.
\newblock {\em Bull. Amer. Math. Soc. (N.S.)}, 27(1):1--67.

\bibitem{ET10}
Ekstrom, E., and Tysk, J. 2010.
\newblock The Black-Scholes equation in stochastic volatility
models.
\newblock {\em Journal of Mathematical Analysis and Applications,}
368:498--507.

\bibitem{Fri64}
Friedman, A. 1964.
\newblock Partial Differential Equations of Parabolic Type.
\newblock {\em  Prentice-Hall Inc., Englewood Cliffs}, N.J.

\bibitem{He93}
Heston, S. 1993.
\newblock A closed-form solution for options with stochastic volatility with applications to bond and currency options.
\newblock {\em Review of Financial Studies,} 6: 327--343.

\bibitem{HW87}
Hull, J. C., and White, A. 1987.
\newblock  The pricing of options on assets with stochastic volatilities.
\newblock {\em The Journal of Finance}, 42:281--300.


\bibitem{Kry80}
 Krylov, N.~V. 1980.
\newblock Controlled diffusion processes. {\em volume~14 of Applications
  of Mathematics}.
\newblock {\em Springer-Verlag}, New York.


\bibitem{LSU67}
Lady{\v{z}}enskaja, O.~A., Solonnikov,  V.~A., and Ural{\cprime}ceva, N.~N. 1967.
\newblock Linear and Quasilinear Equations of Parabolic Type.
\newblock{\em  Translations of Mathematical Monographs, Vol. 23. American
  Mathematical Society}, Providence, R.I..

\bibitem{Lie96}
 Lieberman, G.~M. 1996.
\newblock Second Order Parabolic Differential Equations.
\newblock{\em  World Scientific Publishing Co. Inc.}, River Edge, N.J..

\bibitem{Ok03} 
 {\O}ksendal, B. 2003.
\newblock Stochastic Differential Equations: An Introduction with Applications.
\newblock {\em 6th ed., Springer-Verlag}, Berlin.

\bibitem{OR73}
 Oleinik, O. A., and Radkevich, E. V. 1973.
\newblock Second Order Equations With Nonnegative Characteristic Form.
\newblock {\em Plenum Press}, New York.

\bibitem{Tso85}
 Tso, K. 1985.
\newblock On an {A}leksandrov-{B}akel\cprime man type maximum
principle for
  second-order parabolic equations.
\newblock{\em  Comm. Partial Differential Equations}, 10(5):543--553.


\bibitem{ZC11}
Zhu, S.,  and Chen, W.  2011.
\newblock Should an American option be exercised earlier of later if volatility is not assumed to ba a constant?
\newblock{\em International Journal of Theoretical and Applied Finance,} 14: 1279--1297.

\end{thebibliography}

\def\cprime{$'$}

\end{document}